%% file: main.tex
\documentclass[10pt,a4paper]{article}
\usepackage{amssymb}
\usepackage{enumitem}
\usepackage{amsmath}

\input{macro}

\usepackage{titlesec}

\makeatletter

\def\ps@headings{
	\def\@oddhead{		}
	\def\@evenhead{}
	\def\@evenfoot{\hfil \,\,\thepage\hfill}
	\def\@oddfoot{\hfil \,\,\thepage\hfil}
}

\def\ps@titleheadings{
	\def\@oddhead{}
	\def\@evenhead{}
	\def\@evenfoot{\hfil \,\,\thepage\hfill}
	\def\@oddfoot{\hfil \,\,\thepage\hfil}
}
\makeatother

\title{
	\textbf{StrNim: a variant of Nim played on strings}}
\author{
	Shota Mizuno\thanks{Tohoku University} \and
	Ryo Yoshinaka\footnotemark[1] \and
	Ayumi Shinohara\footnotemark[1]
}
\date{\empty}

\begin{document}
\maketitle

\input{docs/abstract.tex}
\input{docs/introduction.tex}
\input{docs/preliminary.tex}

\input{docs/rules.tex}
\input{docs/result.tex}
\input{docs/conclusion.tex}

\renewcommand{\baselinestretch}{0.70}
\bibliographystyle{plain}
\bibliography{ref}
\end{document}

%% file: macro.tex
\newif\ifdraft \drafttrue   

\ifdraft
\newcommand{\shinocom}[1]{\textcolor{magenta}{[(篠原)#1]}}
\newcommand{\yoshicom}[1]{\textcolor{red}{[(吉仲)#1]}}
\newcommand{\usercom}[1]{\textcolor{blue}{[(user)#1]}}
\newcommand{\todo}[1]{{\color{red}{[ToDo: #1]}}}
\else
\newcommand{\shinocom}[1]{}
\newcommand{\yoshicom}[1]{}
\newcommand{\usercom}[1]{}
\newcommand{\todo}[1]{}
\fi

\usepackage[dvipdfmx]{color}                      
\usepackage[dvipdfmx]{graphicx}                
\usepackage{amsthm}
\usepackage{bm}

\newcommand{\np}{\mathcal{N}\textrm{-position}}
\newcommand{\pp}{\mathcal{P}\textrm{-position}}
\newcommand{\ppos}{\mathcal{P}\mathit{pos}}
\newcommand{\npos}{\mathcal{N}\mathit{pos}}
\newcommand{\str}{s}
\newcommand{\aba}{\texttt{a}^i \texttt{b}^j \texttt{a}^k}
\newcommand{\intpair}[1]{\langle #1 \rangle}
\newcommand{\strnim}{\textsc{StrNim}}
\newcommand{\nim}{\textsc{Nim}}
\newcommand{\tm}[1]{\mathit{TM}_{\!#1}}
\newcommand{\ope}[1]{\stackrel{#1}{\longrightarrow}}
\newcommand{\lose}[2]{Lose({#1})[{#2}]}
\newcommand{\mcal}[1]{\mathcal{#1}}
\newcommand{\mtt}[1]{\mathtt{#1}}

\newcommand{\tsc}[1]{\textsc{#1}}
\newcommand{\next}{\mathit{next}}
\newcommand{\lrfloor}[1]{\left\lfloor{#1}\right\rfloor}

\theoremstyle{definition}
\newtheorem{definition}{Definition}
\theoremstyle{plain}
\newtheorem{theorem}{Theorem}
\newtheorem{corollary}{Corollary}
\newtheorem{lemma}{Lemma}



%% file: docs/abstract.tex
\begin{abstract}
  We propose a variant of \nim, named  \strnim. Whereas a position in \nim \ is a tuple of non-negative integers, that in  \strnim \ is a string, a sequence of characters. 
  In every turn, each player shrinks the string, by removing a substring repeating the same character.
  As a first study on this new game, we present some sufficient conditions for the positions to be $\mcal{P}$-positions.
\end{abstract}

%% file: docs/introduction.tex
\section{Introduction}


\nim{} is a well-known and well-studied combinatorial game.
Combinatorial games consist of a collection of rules, with no element of chance (like dice rolls or card draws) or no hidden information.
Both players have complete knowledge of the game's state at all times.

A position of \nim \ consists of several heaps of tokens.
Two players alternatively remove one or more tokens from any \emph{one} of the heaps, and whoever removes the last token \emph{wins}.
{\nim} is completely analyzed by Bouton~\cite{bouton1901nim}.
A large number of its variants (e.g.,~\cite{wythoffgame,fukuyama2003nim,suetsugu2019delete,locke2021amalgamation,MiyaderaM2023}) have been proposed and studied for more than a hundred years.

In this study, we propose yet another variant of {\nim}, named {\strnim}.
Whereas a position in {\nim} is a tuple of non-negative integers,
that in {\strnim} is a \emph{string}, a sequence of characters.
In every turn, each player \emph{shrinks} the string, by removing a substring repeating the same character.
For instance, for a position \texttt{abbbacc}, the possible successors are 
\texttt{bbbacc}, \texttt{abbacc}, \texttt{abacc}, \texttt{aacc}, \texttt{abbbcc}, \texttt{abbbac}, and \texttt{abbba}.
The \emph{empty string} $\varepsilon$, that is a string of length $0$, is the unique \emph{terminal position} in {\strnim}.
\strnim{} is a generalization of \nim{}.

As a first study of this new game, we show some sufficient conditions for the positions to be $\pp$s, particularly concerning strings whose prooperties are well studied in Stringology~\cite{lothaire1997combinatorics}.


%% file: docs/preliminary.tex
\section{Preliminary}
For a set $X \subseteq \mathbb{N}$ of nonnegative integers, let $\text{mex}(X) = \min(\mathbb{N} \setminus X)$ be the smallest non-negative integer $x$ satisfying $x \notin X$.
For $l,r \in \mathbb{N}$ such that $l \le r$, let $[l,r]$ be the set $\{x \mid x \in \mathbb{N},\, l \le x \le r\}$ of integers in the interval.

For a set $\Sigma$ of characters, $\Sigma^*$ denotes the set of strings over $\Sigma$. 
For a string $s \in \Sigma^*$, 
let $|s|$ denote the length of $s$, and $s^R$ the reverse of $s$.
The $i$-th character of $s$ is denoted by $s[i]$ for $i = 1,2,\ldots,|\str|$, and $\str[i:j]$ denotes the \textit{substring} $s[i]s[i+1] \dots s[j]$, consisting of $s[i]$ through $s[j]$.
We often write a string in \emph{run-length-encoding}, e.g.~$\texttt{aabbba} = \texttt{a}^2\texttt{b}^3\texttt{a}^1$.
\textit{A morphism over} $\Sigma$ is a mapping $h: \Sigma^* \rightarrow \Sigma^*$ satisfying $h(uw) = h(u)h(w)$ for every $u,w \in \Sigma^*$. 
Note that a morphism is uniquely defined by the images $h(a)$ of all $a \in \Sigma$. 
A \emph{coding} is a morphism that satisfies $|h(a)| = 1$ for all $a \in \Sigma$.
For $L \subseteq \Sigma^*$, let $L^*$ be the set of strings consisting of the elements in $L$ connected zero or more times.



%% file: docs/rules.tex
\section{The Rule of StrNim}

\begin{definition}
  (\strnim{}). A position of \strnim{} is a string. Each player, in turn, performs the following operation.
  The player who cannot perform the operation \emph{loses}.
  \begin{description}[topsep=4pt]
  \item[Operation] Select a non-empty substring that is a repetition of a single character, and remove it from the string.
  \end{description}
\end{definition}

For example, the game proceeds as follows, where each player removes the underlined substring.
\[
\mtt{aab\underline{bb}ca} \ope{1}
\mtt{aa\underline{b}ca} \ope{2}
\mtt{aa\underline{c}a} \ope{1}
\mtt{\underline{aaa}} \ope{2} \varepsilon
\]
In this progression, Player~1 loses since no operation is possible on the empty string $\varepsilon$.

For any position $\bm{x} = (x_1,x_2,\cdots,x_n)$ of {\nim}, let us consider a string $\str_{\bm{x}}=c_1^{x_1}c_2^{x_2}\cdots c_n^{x_n}$, where $c_i$'s are mutually distinct characters. 
Then it is easily verified that any possible progression from $\bm{x}$ in \nim{} coincides with that from $s_{\bm{x}}$ in \strnim{} and vice versa.
Therefore, {\strnim} fully subsumes {\nim}.


In contrast, in \strnim{}, two different parts may be occasionally \textit{merged} into one, e.g.~$\mtt{aa\underline{c}a} \rightarrow \mtt{aaa}$, which is the interesting characteristics of \strnim{} that does not happen in {\nim}.
It might appear resembling Amalgamation Nim~\cite{locke2021amalgamation}, a recent variant of {\nim}, where players are allowed to merge two piles into one.
However, these are different games. 
In Amalgamation Nim, players choose whether to remove or merge tokens, whereas in {\strnim}, merging occurs automatically in the process of removing tokens.


%

%% file: docs/result.tex
\section{Results}

In combinatorial games, a $\pp$ is a position such that the \emph{previous} player wins, while an $\np$ is a position such that the \emph{next} player wins.
In {\nim}, a very beautiful characterization is known~\cite{bouton1901nim}; \emph{a position $\bm{x}$ is a $\pp$ iff the binary digital sum of the integers $x_i$'s in $\bm{x}$ is $0$.}
Toward such a goal, we show some initial attempts in the sequel.
We denote the set of $\pp$s and $\mcal{N}$-positions in {\strnim} by $\ppos$ and $\npos$, respectively.
Let $\next(s)$ be the set of strings obtained by one operation on $s$.
By definition, for any $s\in \Sigma^*$, $s \in \ppos$ iff $s' \in \npos$ for all $s' \in \next(s)$.

\begin{lemma}
  \label{next}
  $|next(s)|=|s|$.
\end{lemma}

\begin{proof}

  Let $\str=c_1^{x_1}c_2^{x_2}\cdots c_n^{x_n}$ be any position such that $c_i \neq c_{i+1} \in \Sigma$ and $x_i \geq 1$.
  Then, $\next(s)= \{s_{i,k} \mid 1 \le i \le n \text{ and } 1 \le k \le x_i  \}$ where $s_{i,k}=c_1^{x_1}c_2^{x_2}\cdots  c_i^{x_i-k} \cdots c_n^{x_n}$ is obtained by removing $c_i^k$ from $s$.
  Clearly $s_{i,k}$ are pairwise distinct. Therefore, $|\next(s)|=|s|$.
\end{proof}

For example, $\next(\texttt{abbccb}) = \{\texttt{bbccb}, \texttt{abccb}, \texttt{accb}, \texttt{abbcb}, \texttt{abbb}, \texttt{abbcc}\}$.

\begin{lemma}\label{lem:one}
  If $x c^i y, x c^j y \in \ppos$ with $c \in \Sigma$ and $x,y \in \Sigma^*$, then $i=j$.
\end{lemma}
\begin{proof}
  Suppose $i <j$ and $x c^j y \in \ppos$.
  Then $x c^i y \in \next(x c^j y)$ and thus $x c^i y \notin \ppos$.
\end{proof}

\begin{theorem}
  \label{onlyp}
  Let $s=c_1^{x_1}c_2^{x_2}\cdots c_n^{x_n}$ be any position such that $c_i \neq c_{i+1} \in \Sigma$ and $x_i \geq 1$.
  For any character $c \in \Sigma$ with $c \neq c_n$,
  there exists $y \in [0,|s|]$ such that $s c^{x} \in \ppos$ iff $x=y$.
\end{theorem}
\begin{proof}
By Lemma~\ref{lem:one}, it suffices to show that there exists $x \in [0,|s|]$ such that $s c^{x} \in \ppos$.
To derive a contradiction, assume otherwise. 
For every $i\in [0,|s|]$, since $s  c^i \notin \ppos$, there exists $t_i \in \next(s  c^i) \cap \ppos$.
Since $s c^j \notin \ppos$ for any $j < i$, $t_i$ must be of the form $t_i=s_i  c^i$ for some $s_i \in \next(s)$.
Since $|\next(s)|=|s|$ by Lemma~\ref{next}, there exist two distinct integers $i,j \in [0,|s|]$ such that $s_i = s_j$ by the pigeonhole principle.
It follows that $s_i c^i, s_i c^j \in \ppos$, which contradicts Lemma~\ref{lem:one}. 
\end{proof}
\begin{corollary}
  A position $c_1^{x_1}c_2^{x_2}\cdots c_n^{x_n}$ is an $\mcal{N}$-position if $x_1 + x_2 + \cdots + x_{n-1} < x_n$ holds.
\end{corollary}
For example, $\texttt{abbccaaaaaa} \notin \ppos$.

\subsection{Non-context-freeness}

\begin{theorem}
  The set $\ppos \cap \{ \mtt{a},\mtt{b},\mtt{c}\}^*$ is not context-free.
\end{theorem}
\begin{proof}
  Since context-free languages are closed under intersection with regular sets,
  it is enough to show that $L_1 = \ppos \cap \mtt{a}^*\mtt{b}^*\mtt{c}^*$ is not context-free.
  By Bouton's theorem~\cite{bouton1901nim},
  \[
    L_1 = \{\, \mtt{a}^i \mtt{b}^j \mtt{c}^k \mid i \oplus j \oplus k = 0\,\}\,, 
  \]
  where $\oplus$ denotes the binary digital sum:
  \[
    \Big(\sum_{i=0}^n a_i2^i\Big) \oplus \Big(\sum_{i=0}^n b_i2^i\Big) = \sum_{i=0}^n (a_i \oplus b_i) 2^i
  \]
  holds for $a_0,b_0,\dots,a_n,b_n \in \{0,1\}$ where $\oplus$ is the exclusive disjunction.

  Suppose $L_1$ is context-free.
  Applying the pumping lemma to $w = \mtt{a}^{3 \cdot 2^p} \mtt{b}^{5 \cdot 2^p} \mtt{c}^{6 \cdot 2^p} \in L_1$, we have $u,v,x,y,z \in \Sigma^*$ such that
  \begin{itemize}
  \item $
    xuyvz = \mtt{a}^{3 \cdot 2^p} \mtt{b}^{5 \cdot 2^p} \mtt{c}^{6 \cdot 2^p} 
  $,
  \item $|uyv| < p$,
  \item $|uv| \ge 1$,
  \item $xu^iyv^iz \in L_1$ for all $i \in \mathbb{N}$
  \end{itemize}
  where $p$ is a large number.
  Notice that
  \[
   {\#_\mtt{a}(w)/2^p} = (011)_2\,,\quad
   {\#_\mtt{b}(w)/2^p} = (101)_2\,,\quad
   {\#_\mtt{c}(w)/2^p} = (110)_2\,,
  \] 
  where $\#_a(w)$ denotes the number of occurrences of $a \in\Sigma$ in $w$.
  One can easily see that there are $a,b \in \{\mtt{a},\mtt{b},\mtt{c}\}$ such that $a \ne b$, $u \in a^*$, $v \in b^*$, and $|u|=|v|$.
  Otherwise, $x u^2 y v^2 z \notin L_1$.
  By $|u| < p/2$, one can find an integer $k$ such that $2^p \le k |u| < 2^{p+1}$.
  Let $q = k|u|$.
  Then,
  \[
    \lrfloor{\frac{\#_\mtt{a}(w) + q}{2^p}} = (100)_2,\ 
    \lrfloor{\frac{\#_\mtt{b}(w) + q}{2^p}} = (110)_2,\ 
    \lrfloor{\frac{\#_\mtt{c}(w) + q}{2^p}} = (111)_2.
  \] 
  Let us consider $w' = xu^{k+1}yv^{k+1}z$.
  If $a = \mtt{a}$ and $b = \mtt{b}$, then $\#_\mtt{a}(w') = \#_\mtt{a}(w)+q$,  $\#_\mtt{b}(w') = \#_\mtt{b}(w)+q$,  $\#_\mtt{c}(w') = \#_\mtt{c}(w)$, which implies $w' \notin L_1$.
  If $a = \mtt{a}$ and $b = \mtt{c}$, then $\#_\mtt{a}(w') = \#_\mtt{a}(w)+q$,  $\#_\mtt{b}(w') = \#_\mtt{b}(w)$,  $\#_\mtt{c}(w') = \#_\mtt{c}(w)+q$, which implies $w' \notin L_1$.
  If $a = \mtt{b}$ and $b = \mtt{c}$, then $\#_\mtt{a}(w') = \#_\mtt{a}(w)$,  $\#_\mtt{b}(w') = \#_\mtt{b}(w)+q$,  $\#_\mtt{c}(w') = \#_\mtt{c}(w)+q$, which implies $w' \notin L_1$.
\end{proof}

\input{docs/aba.tex}

\input{docs/anti_palindrome.tex}

\input{docs/constrained.tex}

%% file: docs/aba.tex
\subsection{The form of $\aba$}


This section shows a periodic property of strings in $\mtt{a}^*\mtt{b}^*\mtt{a}^* \cap \ppos$.
Strings of the form $\aba$ are simplest strings that distinguish \strnim{} from \nim{}, where two non-adjacent sections can be merged.\footnote{
  The \strnim{} with positions of the form $\aba$ is equivalent to a specific kind of positions of \tsc{Goishi Hiroi}, which Abuku et al.~\cite{AbukuFKS2024} studied independently.
  We are grateful to them for notifying us of the connection between the two games.
}

For a set $S \subseteq \mathbb{N} \times \mathbb{N}$ of integer pairs and $p \in \mathbb{N}$,
let $S + p \mathbb{N} = \{\intpair{a+kp,b+kp} \mid \intpair{a,b} \in S,\, k \in \mathbb{N} \}$.
We say that a set $T \subseteq \mathbb{N} \times \mathbb{N}$ is \textit{periodic} if there exist two finite subsets $T_1,T_2 \subseteq T$ and an integer $p \geq 1$ such that $T = T_1 \cup (T_2 + p \mathbb{N})$.
We call $p$ the \textit{period} of $T$.
We define $L(j)= \{\intpair{i,k} \mid \aba \in \ppos,\ i \leq k \}$.

\begin{theorem}
  For any $j \geq 1$, $L(j)$ is periodic.
\end{theorem}

\begin{proof}
  We prove this theorem by induction on $j$.
  We first show that $L(1) = \{\intpair{0,1}\} \cup (\{\intpair{2,2}\} + \mathbb{N})$.
  Clearly $\texttt{ba} \in \ppos$. 
  We show $\texttt{a}^i \texttt{ba}^i \in \ppos$ for $i \ge 2$ by induction on $i$.
  If a player takes $\texttt{b}$ in the middle of $\mtt{a}^i \mtt{ba}^i$, the resultant string $\mtt{a}^{2i}$ is obiviously a $\mcal{P}$-position.
  If he removes a substring consisting of $\mtt{a}$, the resultant string is $\mtt{a}^i \mtt{b} \mtt{a}^k$ or $\mtt{a}^k \mtt{b} \mtt{a}^i$ for some $k< i$.
  For the symmetry, we may assume it is the former.
  If $k=1$, the opponent can make it into $\mtt{ba} \in \ppos$.
  If $k>1$, the opponent can make it into $\mtt{a}^{k}\mtt{ba}^{k}$, which is also a $\mcal{P}$-position by induction hypothesis.

  We show that $L(j)$ is periodic for $j \ge 2$.
  Let $\tilde{L}(j) = L(j) \cup \{(k,i) \mid (i,k) \in L(j)\}$.
  It is equivalent for $L(j)$ be periodic and for $\tilde{L}(j)$ be periodic.
  Let $\lose{j}{i}$ denote the value of $k$ such that $\aba \in \ppos$ (which is uniquely determined by Theorem~\ref{onlyp}). Note that $\tilde{L}(j) = \{\intpair{i, Lose(j)[i]} \mid i \ge 0 \}$ holds.
  Define $A(i,j) = \{\lose{j}{t} \mid 0 \le t < i \}$ and $B(i,j) = \{\lose{t}{i} \mid 1 \le t < j \}$.
  Then, $\lose{j}{i} = \text{mex}(A(i,j) \cup B(i,j))$ holds.
  From Theorem~\ref{onlyp}, $\lose{j}{i} \in [i-j,i+j]$.
  We define $A'(i,j) = A(i,j) \cap [i-j,i+j]$, and $B'(i,j) = B(i,j) \cap [i-j,i+j]$
  then the following equation holds.
  \begin{align}
    \label{losemex}
    \lose{j}{i} = \text{mex}([0,i-j-1] \cup A'(i,j) \cup B'(i,j))
  \end{align}
  For a set $S \subseteq \mathbb{N}$ of integers and $p \in \mathbb{N}$,
  let $S + p = \{x + p \mid x \in S \}$.
  By induction hypothesis, for each $t < j$, let $p_t$ be the period of $\tilde{L}(t)$. 
  Let $p'$ be least common multiple of $p_1,p_2,\cdots,p_{j-1}$.
  Since $Lose(t)[i] + p_t = Lose(t)[i+p_t]$,
  there exist integer $s'$ such that for any $i \ge s'$, the following equation holds.
  \begin{align}
    \label{bperiod}
    B'(i,j) +p' = B'(i+p',j)
  \end{align}
  Then, for all $n$, $p = n p'$ and for any $i \ge s$, $B'(i,j) +p = B'(i+p,j)$.
  Since $A'(i,j) \subseteq [i-j,i+j]$, the set $\{A'(i,j) - i \mid i = p'm , m \in \mathbb{N}\} \subseteq [-j,j]$ is finite.
  Thus, there exist integers $n$ and $i_0$ such that $p = n p'$ and the following equation holds.
  \begin{align}
    \label{asame}
    A'(i_0,j)-i_0 &= A'(i_0+p,j) - (i_0 + p) \notag \\
    \Leftrightarrow A'(i_0,j) + p &= A'(i_0+p,j)
  \end{align}
  From Eqs.\ (\ref{losemex}), (\ref{bperiod}), and (\ref{asame}), the following equation holds.
  \begin{align}
    \label{loseperiod}
    Lose(j)[i_0+p] &= \text{mex}([0,i_0+p-j-1] \cup A'(i_0+p,j) \cup B'(i_0+p,j)) \notag \\
    &= \text{mex}([0,i_0+p-j-1] \cup (A'(i_0,j)+p) \cup (B'(i_0,j)+p)) \notag \\
    &= Lose(j)[i_0]+p
  \end{align}
  From (\ref{asame}) and (\ref{loseperiod}), since the following equations hold, $A'(i_0+1,j)+p= A'(i_0+p+1,j)$.
  \begin{align*}
    A'(i_0+1,j) &= A(i_0+1,j) \cap [i_0+1-j,i_0+1+j]\\
    &= (A(i_0,j) \cup \{Lose(j)[i_0]\}) \cap [i_0+1-j,i_0+1+j]\\
    &= (A'(i_0,j) \cup \{Lose(j)[i_0]\}) \setminus \{i_0-j\}
  \end{align*}
  \begin{align*}
    A'(i_0+p+1,j) &= \phantom{(}(A'(i_0+p,j)\phantom{)} \cup \phantom{(}\{Lose(j)[i_0+p]\})\phantom{)} \setminus \{i_0+p-j\} \\
    &= ((A'(i_0,j)+p) \cup (\{Lose(j)[i_0]\}+p)) \setminus \{i_0+p-j\} \\
    &= ((A'(i_0,j) \cup \{Lose(j)[i_0]\}) \setminus \{i_0-j\}) + p\\
    &= A'(i_0+1,j) + p
  \end{align*}
  Thus, for any $i \ge i_0$, $A'(i,j)+p = A'(i+p,j)$ and $Lose(j)[i+p] = Lose(j)[i]+p$.
  Therefore, $\tilde{L}(j) = \{\intpair{i, Lose(j)[i]} \mid i \ge 0 \}$ is periodic.
\end{proof}


\begin{theorem}
  For $1 \leq j \leq 6$, $L(j)$ are given as follows:

  \begin{itemize}[noitemsep]
    \item $L(1) = \{\intpair{0,1}\} \cup (\{\intpair{2,2}\} + \mathbb{N})$

    \item $L(2) = \{\intpair{0,2}, \intpair{1,1}\} \cup (\{\intpair{3,4}\} + 2 \mathbb{N})$

    \item $L(3) = \{\intpair{0,3}, \intpair{1,2}\} \cup (\{\intpair{4,5}\} + 2 \mathbb{N})$

    \item $L(4) = \{\intpair{0,4}, \intpair{1,3}, \intpair{2,5}\} \cup (\{\intpair{6,8}, \intpair{7,9}\} + 4 \mathbb{N})$

    \item $L(5) = \{\intpair{0,5}, \intpair{1,4}, \intpair{2,3}, \intpair{6,9}, \intpair{7,10}, \intpair{8,11}\} \\ \phantom{L(0) = } \cup (\{\intpair{12,14}, \intpair{13,15}\} + 4 \mathbb{N})$

    \item $L(6) = \{\intpair{0,6}, \intpair{1,5}, \intpair{2,4}, \intpair{3,7}, \intpair{8,10}, \intpair{9,11}\} \\ \phantom{L(0) = } \cup (\{\intpair{12,15}, \intpair{13,16}, \intpair{14,17}\} + 6 \mathbb{N})$
  \end{itemize}
\end{theorem}



%% file: docs/anti_palindrome.tex
\subsection{Complementary Palindrome}

\begin{definition}
  A string $\str \in \Sigma^*$ is a \textit{(generalized) complementary palindrome} if $\str = t \cdot f(t^R)$ for some string $t \in \Sigma^*$ and a bijective coding $f$ satisfying $f(a) \neq a$ for any $a \in \Sigma$.
\end{definition}


For example, $\str = \texttt{abccaacb}$ is a complementary palindrome since $s = t \cdot f(t^R)$ for $t = \texttt{abcc}$ and $f(\texttt{a}) = \texttt{b}$, $f(\texttt{b}) = \texttt{c}$, $f(\texttt{c}) = \texttt{a}$.

\begin{theorem}
  Every complementary palindrome is a $\pp$.
\end{theorem}

\begin{proof}
  We show a \emph{mirror strategy}.
  Let $s$ be any complementary palindrome.
  Let $i$ and $j$ be any integers such that $s[i\!:\!j]\! =\! a^k$ for some $a \in \Sigma$ and $k \geq 1$.
  $s[i:j]$ never crosses the center of $s$ since $s[\frac{|s|}{2}] \neq f(s[\frac{|s|}{2}]) = s[\frac{|s|}{2} + 1]$.
  Let $i' = |s| - i + 1$ and $j' = |s| - j + 1$.
  Then $s[j':i']$ never overlaps with $s[i:j]$, and $s[j':i'] = b^k$ for some $b \in \Sigma$ since $f$ is a bijection; $b = f(a)$ if $i < \frac{|s|}{2}$ and $b = f^{-1}(a)$ otherwise.
  For any removal of $s[i:j]$ by a player, the opposite player can always remove $s[j':i']$ so that the resulting string will be a complementary palindrome again.
\end{proof}



%% file: docs/constrained.tex
\subsection{Run length at most one}


\begin{theorem}
  \label{theoremone}
  Let $\str \in \Sigma^*$ be any string with $\str[i] \neq \str[i\!+\!1]$ for $1 \leq i < |\str|$.
  Then $s \in \ppos$ iff $|\str|$ is even.
\end{theorem}

\begin{proof}
  Assume $s$ be a string with $\str[i] \neq \str[i\!+\!1]$ for $1 \leq i < |\str|$ and $|s|$ is even.
  Whatever operation a player performs on $s$, the opposite player can always make it into $s'$ with $s'[i] \neq s'[i\!+\!1]$ for $1 \leq i < |s'|$ and $|s'|$ is even.
  If both sides of a character that a player removes are same, the opposite player should remove one of the same character.
  If both sides of a character that a player removes are different, the opposite player should remove the character at the beginning or end of string.
  If $s$ is a string with $\str[i] \neq \str[i\!+\!1]$ for $1 \leq i < |\str|$ and $|s|$ is odd, $s \notin \ppos$ since a player can remove the first or last character of a string to make it a string with the above condition.
\end{proof}

\subsection{A context-free subset of $\ppos$}

\begin{theorem}
  \label{theoremcontextfree}
  Let $L = \{\mtt{a}^k \mtt{b}^k \mid k \ge 0\} \cup \{\mtt{ba} \}$. $L^* \subseteq \ppos$.
\end{theorem}

\begin{proof}
Assume $s \in L^*$. 
Whatever operation a player performs on $s$, the resulting string is of the form $s_1 \texttt{a}^i \texttt{b}^j s_2$ for some $i,j \ge 0$ with $i \ne j$ and $s_1,s_2 \in L^*$.
Then, the opposite player can always make it into $s_1 \texttt{a}^k \texttt{b}^k s_2 \in L^*$ with $k = \min\{i,j\}$.
\end{proof}

For example, $\texttt{aabb}\cdot\texttt{aaabbb}\cdot\texttt{ba}\in \ppos$ since $\texttt{aabb},\texttt{aaabbb},\texttt{ba} \in L$.
Note that $\texttt{aabb}\cdot\texttt{bbaa} \notin L^*$ since $\texttt{bbaa} \notin L$.


\subsection{Thue--Morse strings}

The \emph{Thue--Morse strings} are well-known in the field of combinatorics on words~\cite{lothaire1997combinatorics,AlloucheS98}.
We analyze the case where the game position is a prefix of the Thue–Morse string.

\begin{definition}
  The $i$-th Thue--Morse strings $\tm{i}$ are recursively defined by $\tm{0} = \texttt{a}$ and $\tm{i} = \tm{i-1} \cdot g(\tm{i-1})$ for $i \geq 1$,
  where $g$ is the coding defined by $g(\texttt{a})=\texttt{b}$ and $g(\texttt{b})=\texttt{a}$.
\end{definition}

For example,
$\tm{1} = \texttt{ab}$,
$\tm{2} = \texttt{abba}$,
$\tm{3} = \texttt{abbabaab}$.
Let $\tm{\infty} = \displaystyle\lim_{n \rightarrow \infty}\!\!\!\tm{n} = \texttt{abbabaabbaababbabaa}\ldots$

\begin{theorem}
  \label{tmstrings}
  $\tm{\infty}[1:i] \in \ppos$ iff $i$ is even.
\end{theorem}

\begin{proof}
Let $t_i = \tm{\infty}[1:i]$.
If $i$ is even, then $t_i \in (\texttt{ab}+\texttt{ba})^* \subseteq L^*$ for $L = \{\texttt{a}^k \texttt{b}^k \mid k \ge 0\} \cup \{\texttt{ba} \}$, which implies $t_i \in \ppos$ by Theorem~\ref{theoremcontextfree}.
If $i$ is odd, $t_i \in (\texttt{ab}+\texttt{ba})^* (\texttt{a}+\texttt{b})$.
By removing the last character of $t_i$, we obtain a string in $L^*$.
\end{proof}

For example, $\tm{\infty}[1:6] = \texttt{abbaba} \in \ppos$ and $\tm{\infty}[1:9] = \texttt{abbabaabb} \notin \ppos$.
Theorem~\ref{tmstrings} completely identifies the $\pp$ and $\np$ when the position is a prefix of $\tm{\infty}$.

%% file: docs/conclusion.tex
\section{Conclusion}

In this paper, we introduced a new game \strnim, a variant of \nim \ played on strings and showed some sufficient conditions for a position in \strnim \ to be in $\ppos$ from various points of view such as stringology. 